\documentclass[letterpaper, 10 pt, conference]{ieeeconf}
\IEEEoverridecommandlockouts
\usepackage{times}
\usepackage{fancyhdr}
\usepackage{float}
\usepackage{cite}
\usepackage{listings}
\usepackage{booktabs}
\usepackage{tabularx}

\usepackage{amsthm}
\usepackage{tikz}
\usepackage{amssymb}
\usepackage{dsfont}
\usepackage{mdwlist}
\usepackage{hyperref}
\hypersetup{colorlinks,linkcolor={blue},citecolor={blue}} 

\usepackage{amsmath,graphicx,xcolor}
\usepackage{subcaption}
\usepackage{subfiles}
\usepackage{import}
\usepackage{dsfont}
\usepackage{xfrac}
\usepackage[font=small,labelfont=bf]{caption}
\usepackage{mathtools}
\usepackage{algorithm}
\usepackage{algpseudocode}
\usepackage{adjustbox}
\usepackage{pifont}
\usepackage{sjmacros}
\usepackage{thmstyles}
\usepackage{stfloats}
\usepackage{cuted}
\usepackage{color}
\usepackage{multirow}
\usepackage{graphicx}
\usepackage{bm}
\usepackage{caption}
\usepackage{subcaption}

\title{\Large \bf Scalable Sensor Scheduling for Continuous-Discrete Kalman Filtering\\ via Information-Form Surrogate Dynamics}
\author{Hyeongmin Choe, SooJean Han
\thanks{*Authors are with School of Electrical Engineering, Korea Advanced Institute of Science \& Technology (KAIST), Daejeon, Republic of Korea. E-mail: \{\texttt{choe5500, soojean\}\texttt{@kaist.ac.kr}}}
}

\begin{document}

\maketitle

\begingroup
\renewcommand\thefootnote{}\footnotetext{%
This work has been submitted to the IEEE for possible publication.
Copyright may be transferred without notice, after which this version
may no longer be accessible.}
\endgroup

\begin{abstract}
% We study sensor scheduling for continuous-discrete Kalman filtering with Poisson measurement arrivals and propose information-form deterministic surrogate for scalable offline design. Unlike the covariance-form surrogate, the sensing rates enter through sensor-specific additive information incremetns, so in the transcribed nonlinear program the mixed state-input derivatives of the dynamics vanish. For any fixed feasible schedule, the inverse information surrogate yields a deterministic lower bound on the conditional-mean covariance, while the covariance-form surrogate yields an upper bound. This gives computable two-sided certificates for the expected estimation cost of the returned schedule under stochastic measurement arrivals. Numerical experiments show substantial reductions in the solver runtime, especially in many-sensor settings, while maintaining comparable realized Monte Carlo performance over the overlapping range.

 We study sensor scheduling for continuous-discrete Kalman filtering with Poisson measurement arrivals and propose an information-form deterministic surrogate for scalable offline design. Unlike the covariance-form surrogate, the sensing rates enter through sensor-specific additive information increments, eliminating mixed state-input derivatives in the transcribed nonlinear program and thereby yielding a simpler derivative structure. We further show that, together with the covariance-form surrogate, the proposed surrogate provides computable two-sided performance bounds for a given schedule under stochastic measurement arrivals. Numerical experiments demonstrate substantial computational savings, especially in many-sensor settings, while retaining comparable realized Monte Carlo performance and providing computable two-sided performance bounds for the returned schedule.
\end{abstract}

\section{Introduction}\label{sec:introduction}

Sensor scheduling allocates sensing and communication resources over time to maintain estimation accuracy under communication and energy constraints. In linear-Gaussian settings, much of the existing literature is formulated in discrete time, with sensing or transmission decisions synchronized to a common clock or periodic schedule~\cite{5765435,vitus2012efficient,mo2011sensor,mo2014infinite,dutta2023unified}. Continuous-time scheduling has also been studied~\cite{leny2011scheduling}, but asynchronous distributed sensing is more realistically described when transmissions are irregular rather than aligned to a common clock. In such settings, a natural estimator model is the \emph{continuous-discrete Kalman filter (CD-KF)}, in which the state evolves continuously while measurements arrive at discrete, irregular times~\cite{shi1999cdkf,jazwinski2007stochastic}.

A closely related recent work is that of Ahdab et al.~\cite{ahdab2025optimal}, who consider Poisson measurement arrivals in a CD-KF setting and derive a covariance-form deterministic surrogate for finite-horizon offline schedule design. The computational cost of that approach, however, depends strongly on the surrogate dynamics used to propagate the estimator state, because that structure determines the derivative couplings in the transcribed nonlinear program and the complexity of the associated Karush--Kuhn--Tucker (KKT) linear systems~\cite{betts1994issues,betts2010practical}. In the covariance form, the sensing rates multiply covariance-dependent Kalman-update terms; after direct transcription, this induces mixed state-control derivative terms and becomes increasingly costly as the number of sensors grows.

Moreover, a surrogate need not coincide with the conditional mean covariance of the underlying stochastic CD-KF. For offline scheduling, where a deterministic policy is optimized in advance but executed under random arrival realizations, it is therefore important to understand the surrogate's optimization structure and how it relates to the exact conditional-mean covariance trajectory and the resulting estimation performance. These motivate the search for a surrogate that is both structurally favorable for optimization and accompanied by a clear comparison to the exact conditional-mean covariance trajectory.

% Our paper
In this paper, we address these issues from the information domain. First, we derive an information-form surrogate for Poisson-rate CD-KF scheduling in which the sensing rates enter additively, avoiding nonlinear state-input coupling and yielding a cleaner derivative structure for gradient-based nonlinear programming. Second, we establish computable two-sided performance bounds for a given schedule by showing optimism of the proposed information-form surrogate and combining it with the conservatism property of the covariance-form surrogate. Numerical experiments show that the resulting method achieves substantial computational gains while providing these performance bounds.

Our paper is organized as follows. Section~\ref{sec:problem_setup} introduces the CD-KF sensor scheduling problem. Section~\ref{sec:method} develops our proposed information-form surrogate and the associated optimal control formulation. Section~\ref{sec:theory} presents a theoretical approximation analysis and deterministic performance-bound results. Section~\ref{sec:experiments} reports numerical experiments.
We conclude the paper in Section~\ref{sec:conclusion}.

\section{Problem Setup}\label{sec:problem_setup}
We consider a distributed sensor network with a global decision-maker that estimates the environment state based on received measurements from the individual sensors in the network.
The environment state is represented by $\xvect(t)\in\mathbb{R}^n$, and satisfies the linear stochastic dynamics
\begin{equation}\label{eq:system_dynamics}
  d\xvect(t) = A \xvect(t)\,dt + \Sigma_w\, dW(t), \hskip.2cm \xvect(0)\sim \mathcal{N}(\mvect_0,P_0),
\end{equation}
where $A\in\mathbb{R}^{n\times n}$, $\Sigma_w\in\mathbb{R}^{n\times r}$, $W(t)$ is an $r$-dimensional standard Wiener process, and $Q\triangleq \Sigma_w\Sigma_w^\top\succeq 0$ is the process noise spectral density.
% Extensions to nonlinear dynamics are possible (e.g., linearization, extended Kalman filter), but are not considered here to keep the paper focused.

Our network consists of $M{\,\in\,}\Nbb$ sensors.
Sensor $j\in\{1,\cdots,M\}$ generates measurements at random times according to a Poisson process $N_j(t)$ with possibly time-varying intensity $\lambda_j(t)\ge 0$.
We assume the transmission processes $\{N_j(t)\}$ are conditionally independent given the intensity trajectories $\{\lambda_j(t)\}$. 
Let $N(t)\triangleq\sum_{j=1}^M N_j(t)$, then at the $k$-th aggregated arrival time $\tau_k \triangleq \inf\{t:\sum_{j=1}^M N_j(t)\ge k\}$, the received measurement is
% Let the aggregated counting process be   $N(t)\triangleq \sum_{j=1}^M N_j(t)$ and $\lambda(t)\triangleq \sum_{j=1}^M \lambda_j(t)$, and define the $n$-th aggregated arrival time $\tau_n\triangleq \inf\{t: N(t)\ge n\}$.
% At each time $\tau_n$, exactly one sensor $j$ triggers (almost surely for Poisson processes); the received measurement is
\begin{equation}\label{eq:sensor_meas}
  \zvect_{j,k} = H_j \xvect(\tau_k) + \vvect_{j,k}, \quad \vvect_{j,k}\sim \mathcal{N}(0,R_j),
\end{equation}
with $H_j\in\mathbb{R}^{p_j\times n}$ and $R_j\in\mathbb{S}^{p_j}_{\succ 0}$.
We further define $\mathcal{F}_t$ to be the filtration generated by all measurements received by the central unit up to time $t$.

For notational simplicity, we present the time-invariant case. The formulation and the theoretical results extend directly to deterministic time-varying coefficients $A(t)$, $\Sigma_w(t)$, $H_j(t)$, and $R_j(t)$, under the same Assumption~\ref{asm:finite_moment}.

% \begin{assumption}\label{assum:detectable}
%     The pair $(A,H_{\text{agg}})$ is detectable, where $H_{\text{agg}}\triangleq [H_1^{\top}\hskip.1cm H_2^{\top}\cdots H_M^{\top}]^{\top}$ is the stack of all sensor measurement matrices.
% \end{assumption}

% \begin{assumption}\label{assum:lambda_regularity}
%     The intensity functions $\lambda_j$ are $\Fcal_t$-predictable and locally integrable in our considered time duration.
% \end{assumption}

\subsection{Background: Continuous-Discrete Kalman Filter}\label{subsec:problem_setup_cdkf}
% We first summarize the main preliminary background relevant for this paper: the continuous-discrete Kalman filter (CD-KF) with Poisson-distributed measurement arrivals.
% We write the conditional mean and covariance as
% \begin{gather*}
%     \mvect(t)\triangleq \mathbb{E}[\xvect(t)\mid\mathcal{F}_t],\\
%     P(t)\triangleq \mathbb{E}\big[(\xvect(t)-\mvect(t))(\xvect(t)-\mvect(t))^\top\mid\mathcal{F}_t\big].
% \end{gather*}
% where we will take $P(0)\equiv P_0\succ 0$.
% By~\eqn{system_dynamics}, we have between measurement arrivals $t\in[\tau_{N(t)},\tau_{N(t)+1})$,
% \begin{subequations}\label{eq:cdkf_predict}
%     \begin{align}
%       \dot \mvect(t) &= A \mvect(t), \label{eq:m_flow}\\
%       \dot P(t) &= A P(t) + P(t)A^\top + Q.
%     \end{align}
% \end{subequations}
% At an arrival time $\tau_n$ from a sensor $j$,
% % applying~\eqn{system_dynamics} and~\eqn{sensor_meas} to 
% the standard Kalman measurement update yields
% \begin{subequations}\label{eq:cdkf_jump}
%     \begin{align}
%       \mvect(\tau_n) &= \mvect(\tau_n^-)+K_j(\tau_n^-)\big(\zvect_{j,n}-H_j \mvect(\tau_n^-)\big), \label{eq:m_jump}\\
%       P(\tau_n) &= P(\tau_n^-) -K_j(\tau_n^-)H_jP(\tau_n^-),\notag
%     \end{align}
% \end{subequations}
% where $K_j(\tau_n^-)\triangleq P(\tau_n^-)H_j^\top\big(H_jP(\tau_n^-)H_j^\top+R_j\big)^{-1}$ is the Kalman gain from sensor $j$ and $(\cdot)^-$ denotes the left limit.
% Additional details and variants of the CD-KF can be found in~\cite{shi1999cdkf,arasaratnam2010cdkf,kulikov2017cdkf} and other references.
Given that our objective is to control estimator accuracy, we begin by summarizing the core principles of the continuous-discrete Kalman filter (CD-KF) with Poisson-distributed arrivals.
We write the conditional error covariance as 
\begin{equation*}
   P(t)\triangleq \mathbb{E}\big[(\xvect(t)-\mvect(t))(\xvect(t)-\mvect(t))^\top\mid\mathcal{F}_t\big],
\end{equation*}
where $\mvect(t)\triangleq \mathbb E[\xvect(t)|\mathcal F_t]$, and we will take $P(0)\equiv P_0\succ 0$.
By~\eqn{system_dynamics}, we have between measurement arrivals $t\in[\tau_{N(t)},\tau_{N(t)+1})$, \begin{equation}\label{eq:cdkf_predict}
    \dot P(t)=AP(t)+P(t)A^\top+Q.
\end{equation}
At an arrival time $\tau_k$ from a sensor $j$, the standard Kalman measurement update yields
\begin{equation}\label{eq:cdkf_jump}
    P(\tau_k)=P(\tau_k^{-})-g_j(P\left(\tau_k^{-})\right),
\end{equation}
where $g_j(X)\triangleq XH_j^\top(H_jXH_j^\top+R_j)^{-1}H_jX$ is the covariance update from sensor $j$ and $(\cdot)^{-}$ denotes the left limit. Additional details and variants of the CD-KF can be found in~\cite{shi1999cdkf,kulikov2017cdkf,arasaratnam2010cdkf} and other references.

\subsection{Sensor Scheduling Problem}\label{subsec:cd_problem_formulation}
In a \emph{CD-KF sensor scheduling problem with Poisson measurement arrivals}, the central processor designs a deterministic open-loop schedule
\begin{equation}
\lamvect(t)=\begin{bmatrix}\lambda_1(t)&\cdots&\lambda_M(t)\end{bmatrix}^\top,\qquad t\in[0,T],
\end{equation}
to minimize an expected estimation cost while satisfying resource constraints. The resulting schedule is then broadcast to the sensor network, and each sensor $j\in\{1,\dots,M\}$ adjusts its transmission rate according to $\lambda_j(t)$.

A standard pathwise estimation cost associated with a covariance trajectory $P(\cdot)$ is
\begin{equation}\label{eq:obj_pathwise}
    J(P)
    \triangleq
    \int_0^T \langle W(t),P(t)\rangle_F\,dt
    +
    \langle W_T,P(T)\rangle_F,
\end{equation}
where $W(t){\,\in\,}\mathbb S_+^n$ and $W_T{\,\in\,}\mathbb S_+^n$ are user-specified weighting matrices. The corresponding control objective is 
\begin{equation}\label{eq:expected_obj}
    \min_{\lamvect\in\mathcal U}\;
    \mathbb E\!\left[J\!\left(P^{\lamvect}\right)\right].
\end{equation}
% The schedule is constrained pointwise in time as
% \begin{equation}\label{ineq:rate_const}
%     C\lamvect(t)\le \mathbf b,\qquad \lamvect(t)\ge_e 0,
% \end{equation}
% for almost every $t{\,\in\,}[0,T]$, where $C{\,\in\,}\mathbb R_+^{n_c\times M}$ is a nonnegative full-row-rank matrix, $n_c\le M$ and $\mathbf b>_e0$, where the subscript $e$ denotes elementwise inequality.
Since the schedule is designed offline at the planning stage, we restrict attention to deterministic open-loop admissible schedules with pointwise linear resource constraints
\begin{align}\label{eq:admissible_set}
    \mathcal U
    \triangleq
    \big\{&
        \lamvect:[0,T]\to\mathbb R^M\textrm{ measurable }:\notag
        \\& C\lamvect(t)\le_e \mathbf b,\;\lamvect(t)\ge_e 0
        \ \text{for a.e. } t\in[0,T]
    \big\},
\end{align}
where $C\in\mathbb R_{+}^{n_c\times M}$, $\mathbf b>_e0$, and each column of $C$ contains at least one strictly positive entry. Here, the subscript $e$ denotes elementwise inequality.
Under this standing assumption, every admissible schedule is componentwise essentially bounded.
Indeed, defining
\begin{equation*}
    \bar \lambda_j \triangleq \min_{i:C_{ij}>0}\frac{b_i}{C_{ij}},\quad j=1,\cdots,M,
\end{equation*}
one has $0\le \lambda_j(t)\le \bar\lambda_j$ for almost every $t\in[0,T]$. This is natural in practice, as it simply requires that every sensor participate in at least one finite resource constraint.

The problem formulation in this section can be viewed as a stochastic continuous-discrete counterpart of sensor scheduling problems studied for the standard discrete-time Kalman filter; see, e.g.,~\cite{dutta2023unified}. The objective in~\eqref{eq:obj_pathwise} covers several practically relevant criteria. For example, choosing $W(t){\,=\,}\theta(t)I$ with $\theta(t)\ge0$ yields a weighted total estimation error criterion, while choosing $W(t){\,\equiv\,}0$ and $W_T{\,=\,}I$ yields a terminal error criterion. Likewise, the linear constraint in~\eqref{eq:admissible_set} subsumes a range of resource models under the rate interpretation. 
% sensor-wise peak-rate constraints ($C{\,=\,}I$), a global sum-budget constraint ($C{\,=\,}\mathbf 1^\top$), and group-coupled constraints.
For instance, when $M=3$, the $C$ matrix
\[
C=
\begin{bmatrix}
1&1&0\\
0&1&1
\end{bmatrix}
\]
models two overlapping local coupling constraints, one for sensors $(1,2)$ and one for sensors $(2,3)$.

\section{Information-form Surrogate}\label{sec:method}

This section develops the information-form representation of the CD-KF and the associated deterministic surrogate dynamics used for scheduling.

\subsection{CD-KF Information Process}
We will first derive the information form of the CD-KF introduced in Section~\ref{sec:problem_setup}.

Since $P_0{\,\succ\,}0$, the continuous covariance propagation preserves positive definiteness on inter-arrival intervals.
The Kalman measurement update in~\eqref{eq:cdkf_jump} also maps $\mathbb S_{++}^n$ into itself. Hence $P(t){\,\in\,}\mathbb S_{++}^n$ almost surely for all $t\in[0,T]$, and the information matrix $Y(t)\triangleq P(t)^{-1}{\,\in\,}\mathbb{S}^n_{\succ 0}$ is well-defined.

From \eqref{eq:cdkf_predict}-\eqref{eq:cdkf_jump} and standard matrix derivative identities, the information matrix satisfies the following dynamics.
Between arrivals, $t\in[\tau_{N(t)},\tau_{N(t)+1})$,
\begin{equation}\label{eq:info_predict}
  \dot Y(t) = -Y(t)A - A^\top Y(t) - Y(t)QY(t),
\end{equation}
and at an arrival time $\tau_k$ from a sensor $j$,
\begin{equation}\label{eq:info_jump}
  Y(\tau_k) = Y(\tau_k^-)+ S_j,\quad S_j\triangleq H_j^\top R_j^{-1}H_j.
\end{equation}
Here, $S_j$ can be thought of as the information increment from sensor $j$. 

\subsection{Information-Form Surrogate Dynamics}\label{subsec:info_rde_surrogate}
To obtain a deterministic surrogate, we replace the random measurement jumps by their rate-weighted average effect. Section~\ref{sec:theory} shows that this can be interpreted as dropping a Jensen gap term from the conditional mean dynamics. We refer to the resulting dynamics as the \emph{information-form surrogate}.

\begin{align}\label{eq:info_rde_mf}
    \dot \ymf(t) =& -\ymf(t)A-A^\top \ymf(t)-\ymf(t)Q\ymf(t)\notag
    \\&+\sum_{j=1}^M \lambda_j(t) S_j,\quad \ymf(0)=Y(0).
\end{align}
To highlight the structure of the proposed dynamics and facilitate the analysis, define $\psi(Y)\triangleq -YA-A^\top Y-YQY.$
   Then the information-form surrogate~\eqref{eq:info_rde_mf} becomes
\begin{equation*}
   \dot\ymf(t)=\psi(\ymf(t))+\sum_{j=1}^M \lambda_j(t)S_j. 
\end{equation*}
Beyond differentiability, the key structural advantage is that the rates enter through sensor-specific matrices $S_j$, rather than covariance-dependent coefficients. This avoids state-dependent control directions in the input channel and simplifies the derivative structure of the discretized nonlinear program.

We therefore solve the deterministic surrogate problem: 
\begin{align}
    \min_{\ymf,\lamvect}\hskip.2cm
    &
    \int_0^T \langle W(t),\ymf(t)^{-1}\rangle_F\,dt
    +\langle W_T,\ymf(T)^{-1}\rangle_F\notag
    \\
    \text{s.t.}\hskip.2cm
    &\dot \ymf(t)=\psi(\ymf(t))+\sum_{j=1}^M \lambda_j(t)S_j
    \hskip.1cm\text{for a.e. }t\in[0,T], \notag\\
    &\ymf(0)=Y(0)\in\mathbb S_{++}^n, \notag\\
    &\lamvect\in\mathcal U, \label{prob:continuous}
\end{align}
where $W(t),W_T$ and $\mathcal U$ are as defined in Section~\ref{sec:problem_setup}.

% Here, on any interval on which $\ymf(t)\in\mathbb{S}_{++}^n$, the inverse-state objective is well defined; throughout the paper we consider feasible trajectories that remain in $\mathbb S_{++}^n$ on $[0,T]$.

% Defining $U_{\lamvect}(t)\triangleq\sum_{j=1}^M \lambda_j(t)S_j$ and $P_{\textrm{info}}(t)\triangleq \ymf(t)^{-1}$, we obtain the Riccati equation $\dot P_{\textrm{info}}=AP_{\textrm{info}}+P_{\textrm{info}}A^\top+Q-P_{\textrm{info}}U_{\lamvect}P_{\textrm{info}},P_{\textrm{info}}(0)\succ0$. Since every admissible schedule $\lamvect\in\mathcal U$ is bounded, standard Riccati theory yields $P_{\textrm{info}}(t)\in\mathbb S_{++}^n$ for all $t\in[0,T]$, therefore the inverse-state objective in~\ref{prob:continuous} is well defined.

For any $\lamvect\in\mathcal U$, the matrix-valued coefficient $U_{\lamvect}(t)\triangleq\sum_{j=1}^M\lambda_j(t)S_j$ is bounded and measurable on $[0,T]$. By defining $P_{\textrm{info}}(t)\triangleq \ymf(t)^{-1}$ we obtain the Riccati equation 
\begin{equation*}
  \dot P_{\textrm{info}}=AP_{\textrm{info}}+P_{\textrm{info}}A^\top+Q-P_{\textrm{info}}U_{\lamvect}P_{\textrm{info}},
\end{equation*}
with $P_{\textrm{info}}(0)\succ0$. Standard Riccati theory yields  $P_{\textrm{info}}(t)\in\mathbb S_{++}^n$ on $t\in[0,T]$, and hence $\ymf(t)\in\mathbb S_{++}^n$ on $[0,T]$. Accordingly, the inverse-state objective in~\eqref{prob:continuous} is well defined.

\section{Theoretical Analysis}\label{sec:theory}

Despite the structural advantages of the information-form surrogate~\eqref{eq:info_rde_mf}, it does not generally coincide with the exact conditional mean trajectory of the CD-KF information process. In this section, we derive complementary comparisons in the information and covariance domains and then combine them into a two-sided bound for the conditional mean covariance trajectory under a given schedule.

Throughout this section, fix an admissible deterministic schedule $\lamvect\in\mathcal U$, and let all exact and surrogate trajectories below be driven by this schedule unless stated otherwise. By the standing assumption on $C$ in Section~\ref{subsec:cd_problem_formulation}, $\lamvect$ is componentwise essentially bounded on $[0,T]$. This makes the following moment condition mild for the CD-KF processes considered here.
\begin{assumption}\label{asm:finite_moment}
    For every $t\in[0,T]$, the random matrices $Y(t)$ and $P(t)=Y(t)^{-1}$ are well defined and satisfy
    \begin{equation*}
        \mathbb E\left[\Vert Y(t)\Vert_F^2+\Vert P(t) \Vert_F \middle| Y(0)\right]<\infty.
    \end{equation*}
\end{assumption}

We define the conditional means
\begin{equation*}
\bar Y(t)\triangleq\mathbb E[Y(t)\mid Y(0)],\quad \bar P(t)\triangleq \mathbb E\left[P(t)\middle| P(0)\right].
\end{equation*}
%  $\bar Y(t)\triangleq\mathbb E[Y(t)\mid Y(0)]$ and $\bar P(t)\triangleq \mathbb E\left[P(t)\middle| P(0)\right]$.
 For every $t\in[0,T]$, both $\bar Y(t)$ and $\bar P(t)$ belong to $\mathbb S_{++}^n$ almost surely, because $Y(t),P(t)\in\mathbb S_{++}^n$ almost surely and positive definiteness is preserved under conditional expectation. Moreover, since $Y(0)=P(0)^{-1}$, the $\sigma$-algebras generated by $Y(0)$ and $P(0)$ coincide, so conditioning on $Y(0)$ or $P(0)$ is interchangeable.

\subsection{Information vs. Covariance Schedule-wise Comparisons}
\begin{lemma}[Jensen-gap Characterization]\label{lem:exact_means}
    The exact conditional means satisfy
    \begin{align}
\dot{\bar Y}(t)=&-\bar Y(t)A-A^\top \bar Y(t)-\bar Y(t)Q\bar Y(t)\notag
\\&+\sum_{j=1}^M \lambda_j(t)S_j-\Delta_Y(t),\label{eq:exact_mean_info}
\\
\dot{\bar P}(t)=&A\bar P(t)+\bar P(t)A^\top+Q\notag
\\&-\sum_{j=1}^M \lambda_j(t)g_j(\bar P(t))-\Delta_P(t),\label{eq:exact_mean_cov}
\end{align}
with $\Delta_Y(t)\succeq 0$ and $\Delta_P(t)\succeq 0$. 
\end{lemma}
Here, $\Delta_Y(t)$ and $\Delta_P(t)$ are Jensen-gap terms, and the jump map $g_j(\cdot)$ is defined in Section~\ref{subsec:problem_setup_cdkf}. The derivation of these identities and Jensen-gap terms is provided in the Appendix. In both domains, the corresponding surrogate is obtained by dropping the Jensen correction from the exact conditional-mean dynamics. Thus, alongside the information-form surrogate~\eqref{eq:info_rde_mf}, we consider the covariance-form surrogate dynamics of~\cite{ahdab2025optimal}, obtained by omitting $\Delta_P(t)$ from~\eqref{eq:exact_mean_cov}:
\begin{equation}
    \dot \pmf(t) = A\pmf(t)+\pmf(t) A^\top+Q-\sum_{j=1}^M\lambda_j(t) g_j\left(\pmf(t)\right).\label{eq:cov_rde_mf}
\end{equation}

The next result shows that these one-sided correction terms induce comparison error dynamics whose negative-semidefinite cone is forward invariant in both the information and covariance domains under a common schedule.
\begin{theorem}\label{thm:complement_comparison}
Let $\ymf(t)$ satisfy~\eqref{eq:info_rde_mf} and $\pmf(t)$ satisfy~\eqref{eq:cov_rde_mf} under the common initial conditions $\ymf(0)=Y(0)=P(0)^{-1}$ and $\pmf(0)=P(0)$.
Then, for all $t\in[0,T]$
    \begin{equation}
        \ymf(t)\succeq \bar Y(t),\quad \pmf(t)\succeq \bar P(t).
    \end{equation}
\end{theorem}
\begin{proof}
    Both comparison proofs are based on the forward invariance of the negative semidefinite cone under suitable matrix error dynamics. Thus, we will repeatedly use the standard tangent-cone characterization
\begin{align}
    &T_{\mathbb S_-^n}(E) =\\
    &\hskip.3cm \Big\{
        Z\in\mathbb S^n:\;
        \vvect^\top Z \vvect\le0\ \text{for all } \vvect\in\ker E,\;E\in\mathbb S_-^n
    \Big\},\notag
\end{align}
together with Nagumo's invariance criterion \cite{aubin2011viability,BLANCHINI19991747}: if $E(0)\in\mathbb S_-^n$ and $\dot E(t)\in T_{\mathbb S_-^n}(E(t))$ for all $t\in[0,T]$, then $E(t)\in\mathbb S_-^n$ for all $t\in[0,T]$.

\textit{(i) Information-form surrogate optimism.}
    Define the error $E_Y(t)\triangleq \bar Y(t)-\ymf(t).$
    Subtracting~\eqref{eq:info_rde_mf} from~\eqref{eq:exact_mean_info} gives
    \begin{equation}
        \dot E_Y
        =
        -E_YA-A^\top E_Y
        -
        \big(\bar YQ\bar Y-\ymf Q\ymf\big)
        -
        \Delta_Y,
        \label{eq:error_info}
    \end{equation}
    with $E_Y(0)=0$.
    We verify that $\mathbb S_-^n$ is forward invariant for~\eqref{eq:error_info}.
    Let $t\in[0,T]$ and suppose $E_Y(t)\preceq0$. For any $\vvect\in\ker E_Y(t)$, we have
    \begin{equation*}
                E_Y(t)\vvect=0
        \quad\Longrightarrow\quad
        \bar Y(t)\vvect=\ymf(t)\vvect.
    \end{equation*}
    Therefore
    $
        \vvect^\top\!\big(-E_YA-A^\top E_Y\big)\vvect=0
    $
    and
    \begin{align*}
        \vvect^\top\!\big(\bar YQ\bar Y-\ymf Q\ymf\big)\vvect
        =&
        (\bar Y \vvect)^\top Q(\bar Y \vvect)
        \\&-(\ymf \vvect)^\top Q(\ymf \vvect)=0.
    \end{align*}
    Since $\Delta_Y(t)\succeq0$, we obtain $\vvect^\top \dot E_Y(t) \vvect
        =
        -\vvect^\top\Delta_Y(t)\vvect
        \le0.$
    Hence $\dot E_Y(t)\in T_{\mathbb S_-^n}(E_Y(t))$ whenever $E_Y(t)\preceq0$. By Nagumo's invariance criterion, $E_Y(t)\preceq0$ for all $t\in[0,T]$. 
    Equivalently, $\ymf(t)\succeq \bar Y(t), $ for all $t\in[0,T]$.
    
\textit{(ii) Covariance-form surrogate conservatism.}
   Defining $E_P(t)\triangleq \bar P(t)-\pmf(t)$. Subtracting~\eqref{eq:cov_rde_mf} from~\eqref{eq:exact_mean_cov} gives
    \begin{equation}
        \dot E_P
        =
        AE_P+E_PA^\top
        -
        \sum_{j=1}^M \lambda_j(t)
        \big(
            g_j(\bar P)-g_j(\pmf)
        \big)
        -
        \Delta_P,
        \label{eq:error_cov}
    \end{equation}
    with $E_P(0)=0.$
    We again verify forward invariance of $\mathbb S_-^n$. Let $t\in[0,T]$ and suppose $E_P(t)\preceq0$. Then $\bar P(t)\preceq \pmf(t)$, and hence
    \begin{equation*}
        H_j\bar P(t)H_j^\top+R_j
        \preceq
        H_j\pmf(t)H_j^\top+R_j.
    \end{equation*}
    Since inversion reverses Loewner order on $\mathbb S_{++}$, it follows that
    \begin{equation*}
        \big(H_j\bar P(t)H_j^\top+R_j\big)^{-1}
        \succeq
        \big(H_j\pmf(t)H_j^\top+R_j\big)^{-1}.
    \end{equation*}
    Now let $\vvect\in\ker E_P(t)$. Then
    \begin{equation*}
        E_P(t)\vvect=0
        \quad\Longrightarrow\quad
        \bar P(t)\vvect=\pmf(t)\vvect=:\wvect,
    \end{equation*}
    and also $\vvect^\top(AE_P+E_PA^\top)\vvect=0$. Therefore,
    \begin{align*}
        \vvect^\top g_j(\bar P(t))\vvect
        &=
        (H_j\wvect)^\top
        \big(H_j\bar P(t)H_j^\top+R_j\big)^{-1}
        (H_j\wvect)
        \\
        &\ge
        (H_j\wvect)^\top
        \big(H_j\pmf(t)H_j^\top+R_j\big)^{-1}
        (H_j\wvect)
        \\
        &=
        \vvect^\top g_j(\pmf(t))\vvect.
    \end{align*}
    Since $\Delta_P(t)\succeq0$, it follows from~\eqref{eq:error_cov} that
    \begin{equation*}
        \vvect^\top \dot E_P(t)\vvect \le 0
        \qquad\text{for all } \vvect\in\ker E_P(t).
    \end{equation*}
    Hence $\dot E_P(t)\in T_{\mathbb S_-^n}(E_P(t))$ whenever $E_P(t)\preceq0$. By Nagumo's invariance criterion, $E_P(t)\preceq0$ for all $t\in[0,T]$. Equivalently, $\bar P(t)\preceq \pmf(t)$ for all $t\in[0,T]$.
\end{proof}

A related covariance-side conservatism statement appears in~\cite{ahdab2025optimal}. However, the proof strategy used there does not directly yield the schedule-wise comparison needed in our setting. We therefore give a separate invariance-based proof in our notation by verifying the tangent-cone condition for the error dynamics as in the information form analysis.

\subsection{Deterministic Bracketing and Objective Bounds}
The previous two comparisons are most useful when combined. The next theorem turns the information-domain optimism and covariance-domain conservatism into a deterministic bracket for the conditional mean covariance trajectory under a fixed schedule.

\begin{lemma}[Conditional Jensen Inequality for Inversion]
    \label{lem:jensen_inverse}
    Let $X$ be an $\mathbb S_{++}^n$-valued random matrix such that $\mathbb E[\|X\|+\|X^{-1}\|\mid \mathcal G]<\infty$ for some sub-$\sigma$-algebra $\mathcal G$. Then
    \begin{equation}
        \big(\mathbb E[X\mid \mathcal G]\big)^{-1}
        \preceq
        \mathbb E[X^{-1}\mid \mathcal G]
        \qquad\text{a.s.}
        \label{eq:jensen_inverse}
    \end{equation}
\end{lemma}

\begin{proof}
    The map $X\mapsto X^{-1}$ is operator convex on $\mathbb S_{++}^n$. Applying conditional Jensen's inequality to the regular conditional law of $X$ given $\mathcal G$ yields~\eqref{eq:jensen_inverse}.
\end{proof}

\begin{theorem}[Schedule-wise Deterministic Bracket]\label{thm:bracketing}
    Let $\ymf(t)$ satisfy~\eqref{eq:info_rde_mf} and $\pmf(t)$ satisfy~\eqref{eq:cov_rde_mf} under the common initial conditions $\ymf(0)=Y(0)=P(0)^{-1}$ and $\pmf(0)=P(0)$.
Then, for all $t\in[0,T]$
    \begin{equation}
        \ymf(t)^{-1}\preceq \bar P(t)\preceq \pmf(t).
    \end{equation}
\end{theorem}

\begin{proof}
    By Theorem~\ref{thm:complement_comparison}, $  \ymf(t)\succeq \bar Y(t)$ and $\pmf(t)\succeq \bar P(t)$. The upper bound is immediate. For the lower bound, since inversion reverses Loewner order on $\mathbb S_{++}^n$, we have $\ymf(t)^{-1}\preceq \bar Y(t)^{-1}$.
    Applying Lemma~\ref{lem:jensen_inverse} with $X=Y(t)$ and $\mathcal G=\sigma(Y(0))$ gives
    \begin{equation*}
                \bar Y(t)^{-1}
        =
        \big(\mathbb E[Y(t)\mid Y(0)]\big)^{-1}
        \preceq
        \mathbb E[Y(t)^{-1}\mid Y(0)].
    \end{equation*}
    Since $Y(t)^{-1}=P(t)$ almost surely and $\sigma(Y(0))=\sigma(P(0))$,
    \begin{equation*}
                \mathbb E[Y(t)^{-1}\mid Y(0)]
        =
        \mathbb E[P(t)\mid P(0)]
        =
        \bar P(t).
    \end{equation*}
    Combining both gives the desired result.
\end{proof}

For planning, the practical value of the bracket is that it yields objective-level bounds. The next corollary transfers the trajectory-wise matrix ordering to the estimation cost and provides deterministic lower and upper bounds on the expected performance of a fixed schedule.
\begin{corollary}[Schedule-wise Performance Bound]
\label{cor:objective_bracketing}
For any admissible deterministic schedule $\lamvect\in\mathcal U$, let
$\ymf^{\lamvect}$ be the information-form surrogate trajectory,
$P^{\lamvect}$ the exact stochastic covariance process, and
$\pmf^{\lamvect}$ the covariance-form surrogate trajectory
driven by the same schedule. Define
    $P_{\mathrm{info}}^{\lamvect}(t)
    \triangleq
    \big(\ymf^{\lamvect}(t)\big)^{-1}$.
Then
\begin{equation}
    J\!\left(P_{\mathrm{info}}^{\lamvect}\right)
    \;\le\;
    \mathbb E\!\left[J\!\left(P^{\lamvect}\right)\mid P(0)\right]
    \;\le\;
    J\!\left(\pmf^{\lamvect}\right).
    \label{eq:objective_bracket}
\end{equation}
\end{corollary}

\begin{proof}
By Theorem~\ref{thm:bracketing}, $   P_{\mathrm{info}}^{\lamvect}(t)
    =
    \big(\ymf^{\lamvect}(t)\big)^{-1}
    \preceq
    \bar P^{\lamvect}(t)
    \preceq
    \pmf^{\lamvect}(t),$
and $\bar P^{\lamvect}(t)\triangleq
    \mathbb E[P^{\lamvect}(t)\mid P(0)]$.
Since $W(t),W_T\succeq0$, Loewner ordering implies
\begin{equation*}
    \langle W(t),P_{\mathrm{info}}^{\lamvect}(t)\rangle_F \leq \langle W(t),\bar P^{\lamvect}(t)\rangle_F \leq \langle W(t),\pmf^{\lamvect}(t)\rangle_F
\end{equation*}
for a.e. $t$, and similarly at the terminal time $T$.
Integrating over $[0,T]$ and using the linearity of conditional expectations yields~\eqref{eq:objective_bracket}. 
\end{proof}

A consequence of Corollary~\ref{cor:objective_bracketing} is that if $P(0)$ is deterministic, then $J(P_{\mathrm{info}}^{\lamvect}) \le \mathbb E[J(P^{\lamvect})] \le J(\pmf^{\lamvect})$.
Thus, we can apply this directly to the planning problem in Section~\ref{sec:problem_setup}, which was posed with deterministic $P_0$. 
The comparison results in this section were stated conditionally so that the arguments remain valid in the slightly more general case of random initial covariance.

\section{Experiments}\label{sec:experiments}

\begin{figure}[t]
    \centering
    \includegraphics[width=0.5\textwidth]{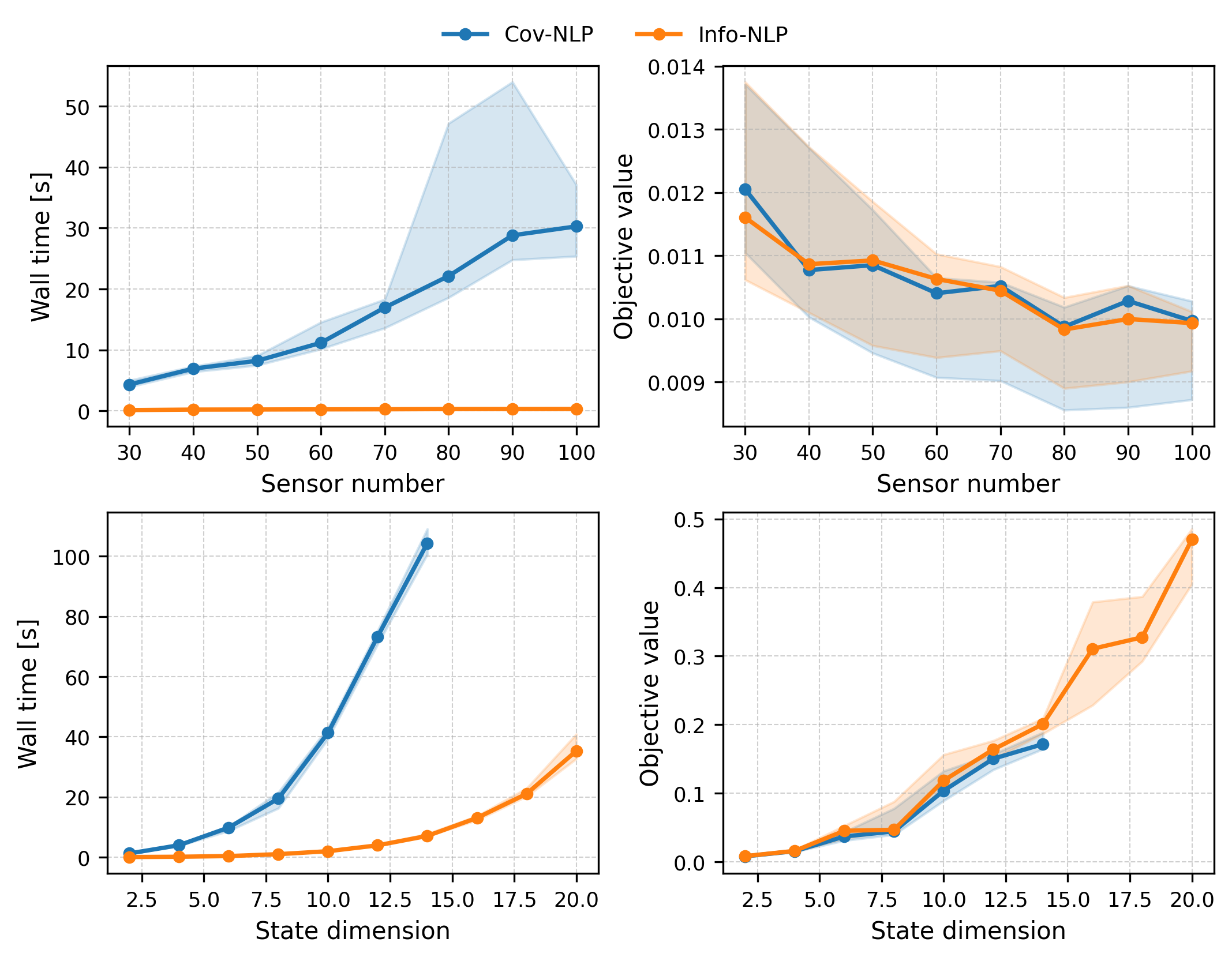}
    \caption{Performance comparison of Cov-NLP and Info-NLP under sensor-number and state-dimension sweeps. Left: IPOPT wall time. Right: MC estimate for the objective. Top: sensor-number scaling. Bottom: state-dimension scaling. Shaded bands denote the interquartile range over 10 repetitions.}
    \label{fig:performance_composite}
\end{figure}

\begin{figure}
    \centering
    \includegraphics[width=1.0\linewidth]{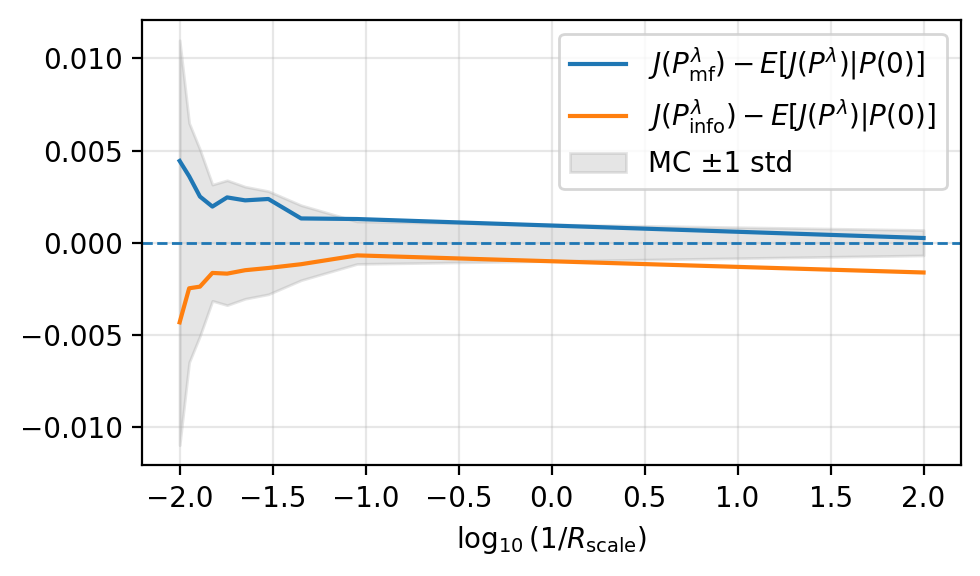}
    \caption{Objective-level validation of Corollary~\ref{cor:objective_bracketing} over the $R$-sweep. The orange and blue curves show the lower- and upper-surrogate objective deviations relative to the MC estimated objective. For illustration, the shaded band shows the Monte Carlo $\pm1$ sample standard-deviation across stochastic arrival realizations.}
    \label{fig:bracket}
\end{figure}

\subsection{Experimental Setup}

We compare \textit{Info-NLP}, based on the information-form surrogate~\eqref{eq:info_rde_mf}, with \textit{Cov-NLP}, based on the covariance-form surrogate~\eqref{eq:cov_rde_mf} (see~\cite{ahdab2025optimal}). Both continuous-time problems are transcribed using the same Euler discretization, Cholesky parametrization, and IPOPT implementation~\cite{wachter2006implementation}, with a $120$\,s time limit. We consider the terminal covariance minimization problem~\eqref{prob:continuous} with $W(t)\equiv 0$ and $W_T=I$ over the horizon $T=3$ using $N=30$ control intervals, under the stagewise budget constraint $\mathbf{1}^\top \lamvect_k \le 5$.

Random plants are generated with a balanced mix of stable and unstable modes. We set $Q=I_n$ and $P_0=100I_n$. For each sensor, $H_j \in \mathbb{R}^{p_j \times n}$ is generated with orthonormal rows, and the measurement noise covariance $R_j$ is sampled as a random positive-definite matrix with eigenvalues drawn uniformly from $[1,10]$ and a random orthogonal eigenbasis.

\textit{Scalability study.}
We vary either the number of sensors ($n=5$, $M \in \{30,40,\dots,100\}$) or the state dimension ($M=30$, $n \in \{2,4,\dots,20\}$), using $10$ random instances per configuration. We report IPOPT wall-clock time and decompose solver effort into two quantities: the number of Newton iterations and the average KKT-Newton assembly time, where the latter is defined as the per-iteration average of the sum of the Hessian-evaluation, Jacobian-evaluation, and residual-construction times. We also report the Monte Carlo (MC) estimated objective value. Each returned schedule is evaluated on the true CD-KF using a finer rollout grid with $N_{\mathrm{eval}}=300$ and $100$ MC runs.

\textit{Performance bound validation.}
We fix a schedule returned by Info-NLP for a reference instance with $n=5$, $M=30$, $p_j=1$, $P_0=100I$, and $Q=I$, and reuse this schedule over a signal-to-noise ratio (SNR) sweep defined by a common scalar measurement-noise scale $R_{\mathrm{scale}} \in [10^{-2},10^{2}]$. For each scaled system, we roll out the same schedule in the information-form surrogate, the covariance-form surrogate, and Monte Carlo CD-KF simulations, and compare the signed deviations in terminal objective relative to the MC estimated objective value.

For both the scalability study and the performance bound validation, we normalize the cost by the initial trace $\operatorname{tr}(P(0))$. This removes trivial growth induced by parameter scaling and makes reductions in terminal estimation error easier to interpret.

\subsection{Numerical Results}

\textit{Scalability and objective value.}
Figure~\ref{fig:performance_composite} summarizes both scalability and MC estimated performance. Info-NLP is consistently faster over the overlapping range, with the advantage being larger in the sensor-number sweep. This trend becomes clearer when these two components are examined separately. In the sensor-number sweep, the median Cov/Info assembly-time ratio grows from $15.80\times$ at $M{\,=\,}30$ to $33.79\times$ at $M{\,=\,}100$ (geometric mean $23.94\times$), whereas the median iteration-count ratio stays between $2.23\times$ and $2.44\times$ (geometric mean $2.32\times$). In the state-dimension sweep, over the valid overlapping range $n \le 14$, the assembly-time ratio decreases from $22.70\times$ to $4.37\times$ (geometric mean $8.28\times$), while the iteration-count ratio remains between $2.55\times$ and $3.30\times$ (geometric mean $2.98\times$). Thus, the wall-clock advantage is largely consistent with the lower per-iteration cost of derivative and KKT construction in the information form, together with a consistent reduction in iteration count.

At the same time, the MC estimated objective gap remains small. In the sensor-number sweep, the normalized realized costs of Cov-NLP and Info-NLP are nearly indistinguishable. In the state-dimension sweep, the gap remains modest over the overlapping range and becomes visible only at the high-dimensional end, where Cov-NLP achieves a slightly lower MC estimated objective value. This is consistent with the theoretical roles of the two surrogates: Info-NLP optimizes an optimistic lower surrogate, whereas Cov-NLP optimizes a conservative upper surrogate.

\textit{Performance bound validation.}
Figure~\ref{fig:bracket} shows objective-level validation of Corollary~\ref{cor:objective_bracketing} over a wide SNR range. Across the full sweep, the MC estimated objective value lies within the two-sided deterministic bound, and the normalized two-sided bound width remains below roughly $5 \times 10^{-3}$. Because the bounds are computed for the same schedule returned by Info-NLP, the objective predicted by Info-NLP can be interpreted as a deterministic lower bound for that feasible schedule, while the covariance-form rollout provides a deterministic upper bound for the same schedule.

\section{Conclusion}\label{sec:conclusion}
We proposed an information-form surrogate for sensor scheduling in continuous-discrete Kalman filtering with Poisson measurement arrivals. The proposed formulation yields a scalable surrogate for direct transcription and, through the accompanying theoretical analysis, provides computable two-sided performance bounds for a planned schedule under stochastic measurement arrivals. Numerical results show substantial computational savings over the covariance-form surrogate across overlapping problem scales, while retaining comparable realized performance and providing interpretable deterministic bounds on the stochastic CD-KF objective. Future work will consider extending this framework to richer sensing models and broader classes of scheduling constraints.

\section*{Acknowledgments}
This research was supported by the Air Force Office of Scientific Research under award number FA23862514038.

\appendix
\begin{proof}[Proof of Lemma~\ref{lem:exact_means}]\label{proof:lemma1}
For the information process, let $\psi(Y)\triangleq -YA-A^\top Y-YQY$, $\phi_j(Y)\triangleq Y+S_j.$

Since $Y(t)$ is a time-inhomogeneous Markov jump process with drift $\psi$ and jump map $\phi_j$, Dynkin's formula applied entrywise to the linear test function $f(Y)=Y$ gives
\begin{align*}
    \dot{\bar Y}(t)
=&
-\bar Y(t)A-A^\top\bar Y(t)
-\mathbb E[Y(t)QY(t)\mid Y(0)]
\\&+\sum_{j=1}^M \lambda_j(t)S_j.
\end{align*}
Now write $Y(t)=\bar Y(t)+\widetilde Y(t)$ with $\widetilde Y(t)\triangleq Y(t)-\bar Y(t)$. Since $\mathbb E[\widetilde Y(t)\mid Y(0)]=0$, $\mathbb E[Y(t)QY(t)\mid Y(0)]
=
\bar Y(t)Q\bar Y(t)
+
\mathbb E[\widetilde Y(t)Q\widetilde Y(t)\mid Y(0)].$
Defining
\begin{equation*}\Delta_Y(t)\triangleq
    \mathbb E[\widetilde Y(t)Q\widetilde Y(t)\mid Y(0)],
\end{equation*}
we have $\Delta_Y(t)\succeq 0$ because $Q\succeq 0$, which yields~\eqref{eq:exact_mean_info}.

For the covariance process, let $\psi_P(P)\triangleq AP+PA^\top+Q,$  $\phi_{P,j}(P)\triangleq P-g_j(P).$ By the Woodbury identity, $P-g_j(P)=(P^{-1}+S_j)^{-1}\succeq 0,$ so in particular \(0\preceq g_j(P)\preceq P\), and the integrability required by Dynkin's formula follows from Assumption~\ref{asm:finite_moment}. Applying Dynkin's formula entrywise to \(f(P)=P\) gives
\begin{equation}
    \dot{\bar P}(t)=A\bar P(t)+\bar P(t)A^\top+Q-\sum_{j=1}^M \lambda_j(t)\mathbb E[g_j(P(t))|P(0)].
\end{equation}

Since each $g_j$ is Loewner-convex~\cite[Appendix A]{ahdab2025optimal}, $\mathbb E[g_j(P(t))\mid P(0)]\succeq g_j(\bar P(t))$, by conditional Jensen's inequality. Defining
\begin{equation*}
    \Delta_P(t)\triangleq
\sum_{j=1}^M \lambda_j(t)
\Big(
\mathbb E[g_j(P(t))\mid P(0)]-g_j(\bar P(t))
\Big),
\end{equation*}
we obtain $\Delta_P(t)\succeq 0$ and hence~\eqref{eq:exact_mean_cov}.
\end{proof}

\bibliographystyle{IEEEtran}
\bibliography{%
refs%
}

\end{document}